\documentclass[a4paper,10pt]{article}
\usepackage{latexsym,amsthm,amsmath,amssymb,url,textcomp}
\usepackage{amsfonts}
\usepackage{graphicx}
\usepackage{subfigure}
\usepackage[USenglish]{babel}

\newtheorem{lemma}{Lemma}

\newtheorem{theorem}{Theorem}

\newtheorem{claim}{Claim}

\title{Note on Large Subsets of Binary Vectors with Similar Distances}
\author{Gregory Gutin and Mark Jones\\Royal Holloway, University of London\\Egham, Surrey, TW20 0EX, UK}

\begin{document}

\maketitle

\begin{abstract}
We consider vectors from $\{0,1\}^n$. The weight of such a vector $v$ is the sum of the coordinates of $v$.
The distance ratio of a set $L$ of vectors
is ${\rm dr}(L):=\max \{\rho(x,y):\ x,y \in L\}/ \min \{\rho(x,y):\ x,y \in L,\ x\neq y\},$ where $\rho(x,y)$
is the Hamming distance between $x$ and $y$. We prove that (a) for every constant $\lambda>1$ there are no positive constants $\alpha$ and $C$ such that every set $K$ of at least $\lambda^p$ vectors with weight $p$ contains a subset $K'$ with $|K'|\ge |K|^{\alpha}$ and ${\rm dr}(K')\le C$,
 (b) For a set $K$ of vectors with weight $p$, and a constant $C>2$, there exists $K'\subseteq K$ such that ${\rm dr}(K')\le C$ and $|K'| \ge |K|^\alpha$, where $\alpha = 1/ \lceil \log(p/2)/\log(C/2) \rceil$.
\end{abstract}

\section{Introduction}

We will consider $n$-dimensional binary vectors (i.e., vectors from $\{0,1\}^n$) and call them {\em $n$-vectors}.
The {\em (Hamming) weight} $|v|$ of an $n$-vector $v$ is the sum of the coordinates of $v$. The {\em (Hamming) distance} $\rho(u,v)$ between $n$-vectors $u,v$ is the number of coordinates where $u$ and $v$ differ. The {\em distance ratio} of a set $L$ of $n$-vectors
is $${\rm dr}(L):=\frac{\max \{\rho(x,y):\ x,y \in L\}} {\min \{\rho(x,y):\ x,y \in L,\ x\neq y\}}.$$

Let $p\le n$ be positive integers. Abramovich and Grinshtein \cite{AbrGri} asked whether the following claim holds true:

\begin{claim}\label{cl}
There exist positive constants $\alpha$, $C$ and $\lambda>1$ such that every  set $K$ of at least $\lambda^p$
$n$-vectors with Hamming weight $p$ contains a subset $K'$ with $|K'|\ge |K|^{\alpha}$ and ${\rm dr}(K')\le C$.
\end{claim}

If the claim is true, it can be used in statistics for establishing the lower bounds for the minimax risk of estimation in various sparse settings  \cite{AbrGri,RigTsy11}. If the claim is not true, any counterexample can be used to impose some conditions on $K$ such that the claim becomes true and, thus, useful for establishing the lower bounds.

The following example shows that for some sets $K$ the claim is true.
Let $p<n/2$ and let $\Omega$ denote the set of all $n$-vectors of weight $p$. By Lemma A.3 in \cite{RigTsy11} (which is a generalization of the Varshamov-Gilbert lemma attributed to Reynaud-Bouret \cite{rey03}), there exists a subset $\Omega'$ of $\Omega$ such that $\rho(x,y)\ge (p+1)/4$ for all distinct $x,y\in \Omega'$ and $|\Omega'|\ge (1+en/p)^{\beta p}$ for some $\beta\ge 9\cdot 10^{-4}$. It follows that ${\rm dr}(\Omega')<8$ (since $\rho(u,v)\le 2p$ for all $u,v\in \Omega$). Moreover, since $|\Omega|={n \choose p}<(en/p)^p$ and $|\Omega'|> (en/p)^{\beta p}$, we have $|\Omega'|>|\Omega|^{\beta}.$

Unfortunately, in general, the claim is not true and we give a counterexample to the claim in Section \ref{sec:ex}. In Section \ref{sec:pos}, we show that a weaker claim holds: there exists $K'\subseteq K$ such that ${\rm dr}(K)\le C$ and $|K'| \ge |K|^\alpha$,
where $\alpha = 1/ \lceil \log(p/2)/\log(C/2) \rceil$ ($C>2$). We conclude the paper with an open problem stated in Section \ref{sec:dis}.

\medskip

Henceforth $[s]:=\{1,\ldots , s\}$ for a positive integer $s$.

\section{Counterexample}\label{sec:ex}

Let us fix constants $C\ge 1$, $0<\alpha\le 1$ and $\lambda>1$. We will show that
there is no set $K$ of at least $\lambda^p$ $n$-vectors satisfying Claim \ref{cl} for these $C$ and $\alpha$. In this section, we will use
fixed positive integers $t,a,p,q$ and $n$ satisfying the following:
\begin{enumerate}
 \item $1/t< \alpha$,  $a > C$;
 \item $p$ is a multiple of $a^t$;
 \item $q^t \ge \lambda^p$;
 \item $n \ge p+p(q-1)\sum_{j=1}^{j=t}(q/a)^{t-j}$.
\end{enumerate}

We say a set $L$ of $n$-vectors
is a {\em ${\cal C}_0$-set}
if $L$ consists of a single vector.
For $i\in [t]$, a set $L$ of vectors
is a {\em ${\cal C}_i$-set}
if it satisfies the following:
\begin{enumerate}
\item $|L|=q^i;$
 \item $\max \{\rho(x,y):\ x,y \in L\} = 2p/a^{t-i};$
 \item $L$ can be partitioned into $q$ sets $L_1, \ldots , L_q$ such that for each $r$, $L_r$
is a ${\cal C}_{i-1}$-set,
and for all $x \in L_r, y \in L_s$ with $r \neq s$, $\rho(x,y)= 2p/a^{t-i}$.
\end{enumerate}

\begin{lemma}\label{lem1}
For each $i\in [t]$, there is a set $K$ of $n$-vectors
with Hamming weight $p$
such that $K$ is a ${\cal C}_i$-set.
\end{lemma}
\begin{proof}
For a set $L$ of $n$-vectors
with Hamming weight $p$
to be a ${\cal C}_i$-set,
we need that  $$\max \{\rho(x,y):\ x,y \in L\} = 2p/a^{t-i}.$$
So for every pair $x,y \in L$ of distinct $n$-vectors, there must be a set $X \subseteq[n]$ with $|X|\ge p-p/a^{t-i}$,
such that $x_r=y_r=1$ for all $r \in X$.
In fact, in our construction below we will assure that in a
${\cal C}_i$-set,
there exists $X \subseteq[n]$ with $|X|\ge p-p/a^{t-i}$ such that $x_r=1$ for all $x \in L$ and $r \in X$.

For some $S \subseteq T \subseteq [n]$, we say a set $L$ of $n$-vectors is a ${\cal C}_i$-set {\em between $(S,T)$} if $L$ is
a ${\cal C}_i$-set,
and for all $x \in L$, $x_r=1$ if $r\in S$ and $x_r=0$ if $r \notin T$.
We give a recursive method to construct a ${\cal C}_i$-set between $(S,T)$ when $|S|=p-p/a^{t-i}$
and $|T|$ is large enough (we calculate the required size of $T$ later).
We can then construct the required set $K$ by constructing a ${\cal C}_t$-set between $(\emptyset, [n])$.

Given $S, T$, construct a ${\cal C}_i$-set $L$ between $(S,T)$ as follows.
If $i=0$, return a single $n$-vector $x$ of Hamming weight $p$, such that $x_r=1$ for all $r \in S$ and $x_r=0$ for all $r \notin T$.

If $i\ge 1$,
partition $T\backslash S$ into $q$ sets $T_1, \dots , T_q$, such that
$-1\le |T_r|-|T_s|\le 1$ for all $r,s$.
For each $1\le r \le q$, let $S_r$ be a subset of $T_r$ of size $p/a^{t-i} - p/a^{t-(i-1)}$. Then for each $r$ construct a ${\cal C}_{i-1}$-set $L_r$ between $(S\cup S_r, S\cup T_r)$, and let $L$ be the union of these sets.
(Note that $|S\cup S_r|=p-p/a^{t-(i-1)}$, as required for the recursion.)

Observe that since $|S|=p-p/a^{t-i}$, $\max \{\rho(x,y):\ x,y \in L\} \le 2p/a^{t-i}$.
Furthermore, since $T_1, \dots, T_q$ are disjoint,
for $x \in L_r, y \in L_s$ with $r \neq s$, $\rho(x,y)= 2p/a^{t-i}$.
Finally note that $|L|=\sum_{r=1}^{r=q}|L_r| = qq^{i-1}=q^i$. Therefore $L$ satisfies all the conditions of a ${\cal C}_i$-set between $(S,T)$.

We now calculate a bound $f_i$ such that we can construct a ${\cal C}_i$-set between $(S,T)$ when $|S|=p-p/a^{t-i}$ as long as $|T| \ge f_i$.

Clearly $f_0=p$.
For $i> 0$, in the construction above we require that $|S\cup T_r|\ge f_{i-1}$ for each $1 \le r \le q$. Therefore we require $$f_i = |S| + q(f_{i-1}-|S|) = qf_{i-1} - (q-1)(p-p/a^{t-i}).$$
Observe that this is satisfied by setting $f_i=p+p(q-1)\sum_{j=1}^{j=i}(q^{i-j}/a^{t-j})$.

So to construct a ${\cal C}_i$-set between $(\emptyset, [n])$, it suffices that $n \ge p+p(q-1)\sum_{j=1}^{j=i}(q^{i-j}/a^{t-j})$, which holds by Part 4 of the conditions on $t,a,p,q$ and $n$ given in the beginning of this section.
\end{proof}


\begin{theorem}\label{th1}
There is a set $K$ of $n$-vectors for which Claim \ref{cl} does not hold.
\end{theorem}
\begin{proof}
We will construct a set $K$ such that for any subset of $K$ with more than $q=|K|^{1/t}$ vectors, the distance ratio  is at least $a$. This implies that for any subset with at least $|K|^{\alpha}$ vectors the distance ratio is greater than $C$, as required.

By Lemma \ref{lem1}, we may assume that we have a
${\cal C}_i$-set $K$.
Thus, $K$ can be partitioned into $q$ sets $K_1, \ldots, K_q$ such that for each $r$,
$K_r$ is a ${\cal C}_{i-1}$-set,
and for all $x \in K_r, y \in K_s$ with $r \neq s$, $\rho(x,y)= 2p/a^{t-i}$.

Note that any subset $K'\subseteq K$ of more than $q$ vectors will contain at least two vectors from $K_r$ for some $r$ and so $\min \{\rho(x,y):\ x,y \in K', x\neq y\} \le 2p/a^{t-i+1}$; furthermore if $K'$ contains vectors from $K_r$ and $K_s$ for $r\neq s$ then $\max \{\rho(x,y):\ x,y \in K'\} \ge 2p/a^{t-i}$.

Therefore, for any $K' \subseteq K$ with $|K'|>q$, either ${\rm dr}(K')\ge a$, or $K'\subseteq K_i$ for some
${\cal C}_{i-1}$-set $K_i$.
Furthermore there is no $K'\subseteq K$ with $|K'|>q$ if
$K$ is a ${\cal C}_1$-set.
So by induction on $i\ge 1$, every $K'\subseteq K$ with $|K'|>q$ has  ${\rm dr}(K')\ge a$.
By letting $i=t$, we complete the proof of the theorem.
\end{proof}

\section{Positive Result}\label{sec:pos}

Given a set $K$ of $n$-vectors, we are interested in finding a subset $K'\subseteq K$ as large as possible such that ${\rm dr}(K')\le C$, for some constant $C$.
The following is such a result.

\begin{theorem}\label{th2}
Let $K$ be a set of $n$-vectors with Hamming weight exactly $p$, and let $C> 2$ be a constant. Then there exists $K'\subseteq K$ such that ${\rm dr}(K')\le C$ and $|K'| \ge |K|^\alpha$, where $\alpha = 1/ \lceil \log(p/2)/\log(C/2) \rceil$.
\end{theorem}
\begin{proof}
Let $t = \lceil \log(p/2)/\log(C/2) \rceil = 1/\alpha$.

Let $K_1 = K$. For each $1 \le i < t$, let $K_{i+1}$ be a maximal subset of $K_i$ such that $\min \{\rho(x,y) | x,y \in K_{i+1},\ x\neq y\} \ge C^i/2^{i-1}$.
For each vector $z \in K_{i}$, let $N_i(z)$ be the set of vectors $x \in K_i$ for which $\rho(x,z) \le C^i/2^{i-1}.$

Observe that $\max \{\rho(x,y) | x,y \in N_i(z)\} \le C^i/2^{i-2}$. Since
$\min \{\rho(x,y) | x,y \in K_i\} \ge C^{i-1}/2^{i-2}$, it follows that ${\rm dr}(N_i(z))\le C$.
Note furthermore that by the maximality of $K_{i+1}$, every vector in $K_i$ is in $N_i(x)$ for some $x \in K_{i+1}$.
Therefore, for $1 \le i < t$, we either have that $|N_i(x)| \ge |K|^\alpha$ for some $x \in K_{i+1}$, in which case we are done, or $|K|^\alpha|K_{i+1}|\ge |K_i|$.
By induction, we have that $|K_i|\ge |K|/|K|^{\alpha(i-1)}$ for $1 \le i \le t$ (or else we can find a set $N_i(x)$
satisfying the theorem).
In particular, we have that $|K_t|\ge |K|/|K|^{\alpha(t-1)} = |K|/|K|^{1-\alpha} = |K|^{\alpha}$.

Now observe that $\max \{\rho(x,y) | x,y \in K_t\} \le 2p$. Furthermore, $$\min \{\rho(x,y) | x,y \in K_t,\ x\neq y\} \ge C^{t-1}/2^{t-2} = (4/C)(C/2)^t \ge (4/C)(p/2) = 2p/C. $$ Therefore ${\rm dr}(K_t) \le C$.
This completes the proof.
\end{proof}

\section{Discussion}\label{sec:dis}

Let $p$ and $n$  be integers ($n>p>0$), and let $C\ge 1$ be a real.
We can view our results as a study of a maximal positive-valued function $\alpha(C,p,n)$ defined as follows.
For $\lambda$ large enough, every set $K$ of at least $\lambda^p$ $n$-vectors with Hamming weight exactly $p$, there exists $K'\subseteq K$ such that ${\rm dr}(K')\le C$ and $|K'| \ge |K|^{\alpha(C,p,n)}$.
In Theorem \ref{th1}, we prove that for every positive real $\alpha_0$ and $C\ge 1$ there exists $p$ such that for $n$ large enough
$\alpha(C,p,n)<\alpha_0$. In Theorem \ref{th2}, we show that $\alpha(C,p,n) \ge 1/\lceil \log(p/2)/\log(C/2) \rceil$ provided $C>2$. It would be interesting to improve this bound on $\alpha(C,p,n)$.

\end{document}